\documentclass[a4paper,12pt]{article}
\usepackage[14pt]{extsizes} 
\usepackage[nottoc]{tocbibind}
\usepackage[mathscr]{eucal}
\usepackage{amsmath,amsfonts,amssymb,amsthm}
\usepackage{times}
\usepackage{hyperref}

\numberwithin{equation}{section} %\makeatletter
\newtheorem{thm}{Theorem}[section]

 \newtheorem{prop}[thm]{Proposition}
 
 \newtheorem{definition}[thm]{Definition}
 
\renewcommand{\simeq}{\cong}
\newcommand{\bref}[1]{\textbf{\ref{#1}}}

\newcommand{\im}{\mathop{\mathrm{Im}}}

\newcommand{\gh}[1]{\mathrm{gh}(#1)}

\newcommand{\dv}{\mathrm{d_v}}
\newcommand{\dx}{\mathrm{d}_X}

\renewcommand{\d}{\partial}
\renewcommand{\dh}{\mathrm{d_h}}

\renewcommand{\geq}{\,{\geqslant}\,}
\renewcommand{\leq}{\,{\leqslant}\,}

\newcommand{\binner}[2]{%
  {\langle}\kern-4.15pt{\langle}#1{,}\,#2{\rangle}\kern-4.15pt{\rangle}}
\newcommand{\commut}[2]{[#1{,}\,#2]}

\newcommand{\ab}[2]{\big(#1,#2\big)}

\newcommand{\half}{\mathchoice{%
    \ffrac{1}{2}}{\frac{1}{2}}{\frac{1}{2}}{\frac{1}{2}}}

\newcommand{\hhalf}{\frac{1}{2}}

\newcommand{\ffrac}[2]{\raisebox{.5pt}%
  {\footnotesize$\displaystyle\frac{#1}{#2}$}\kern1pt}

\newcommand{\dl}[1]{\mathchoice{\ffrac{\d}{\d #1}}{\frac{\d}{\d #1}}{\ffrac{\d}{\d #1}}{\ffrac{\d}{\d #1}}}

\newcommand{\Liealg}{\mathfrak} 
\newcommand{\algg}{\Liealg{g}}

\newcommand{\algA}{\mathcal{A}}
\newcommand{\cC}{\mathcal{C}}
\newcommand{\fC}{\mathbb{C}}

\newcommand{\fZ}{\mathbb{Z}}

 \def\cE{\mathcal{E}}

 \def\cI{\mathcal{I}}
 
\def\cK{\mathcal{K}}
 \def\cL{\mathcal{L}}
 \def\cX{\mathcal{X}}
\newcommand\blfootnote[1]{%
  \begingroup
  \renewcommand\thefootnote{}\footnote{#1}%
  \addtocounter{footnote}{-1}%
  \endgroup
}

\usepackage[left=3cm,right=2cm,top=3cm,bottom=3cm,bindingoffset=0cm]{geometry}
\UseRawInputEncoding
\usepackage{autobreak,mathtools}
\usepackage{authblk}
\usepackage{tikz-cd}
\usepackage{stmaryrd}

\usepackage{lipsum}
\usepackage{tikz-cd}

\newcommand{\wker}{{\widetilde{\cK}}}
\newcommand{\pdv}[1]{\frac{\d}{\d #1}}
\newcommand{\h}{\cL}
\newcommand{\order}[1]{\mathcal{O}(#1)}
\newtheorem{proposition}[thm]{Proposition}

\title{Consistent deformations in the presymplectic BV-AKSZ approach}

\author[1]{Jordi Frias}

\author[2,$\dagger$,$\ddagger$]{~Maxim Grigoriev }

\affil[1]{\textsl{ Institute for Theoretical and Mathematical Physics,\protect\\
  Lomonosov Moscow State University, 119991 Moscow, Russia  \vspace{5pt}}}

\affil[2]{\textsl{Service de Physique de l'Univers, Champs et Gravitation, \protect\\ Universit\'e de Mons, 20 place du Parc, 7000 Mons, 
Belgium \vspace{5pt}}}

\date{}

\begin{document}

\maketitle

\begin{abstract}
We develop a framework for studying consistent interactions of local gauge theories, which is based on the presymplectic BV-AKSZ formulation. The advantage of the proposed approach is that it operates in terms of finite-dimensional spaces and avoids working with quotient spaces such as local functionals or functionals modulo on-shell trivial ones. The structure that is being deformed is that of a presymplectic gauge PDE, which consists of a graded presymplectic structure and a compatible odd vector field. These are known to encode the Batalin--Vilkovisky (BV) formulation of a local gauge theory in terms of the finite dimensional supergeometrical object. Although in its present version the method is limited to interactions that do not deform the presymplectic structure and relies on some natural assumptions, it gives a remarkably simple way to analyse consistent interactions. The approach can be considered as the BV-AKSZ extension of the frame-like approach to consistent interactions. We also describe the underlying homological deformation theory, which turns out to be slightly unusual compared to the standard deformations of differential graded Lie algebras. As an illustration, the Chern-Simons and YM theories are rederived starting from their linearized versions.
\end{abstract}

\blfootnote{${}^{\dagger}$ Supported by the ULYSSE Incentive
Grant for Mobility in Scientific Research [MISU] F.6003.24, F.R.S.-FNRS, Belgium.}
%\blfootnote{${}^{*}$ {Corresponding Author: grigoriev.max@gmail.com}}

\blfootnote{${}^{\ddagger}$ Also at Lebedev Physical Institute and Institute for Theoretical and Mathematical Physics, Lomonosov MSU}

\newpage

\tableofcontents

\section{Introduction}
The usual approach to consistent interactions of local gauge theories is based on the celebrated Noether procedure which amounts to simultaneous consistent deformation of the Lagrangian and gauge generators. The difficulty of the procedure has to do with the fact that even cubic vertexes are defined modulo on-shell vanishing terms  and total derivatives, bringing serious technical difficulties~\cite{Berends:1984rq}.

The Noether procedure has been cast into the well-defined homological deformation theory by Barnich and Henneaux~\cite{Barnich:1993vg} who rephrased it in terms of the Batalin-Vilkovisky (BV)~\cite{Batalin:1981jr,Batalin:1983wj} formulation of the underlying gauge theory. In particular, it was shown that the deformation procedure is controlled by the BRST cohomology in the space of local functionals while possible  obstructions are controlled by the Lie bracket induced by the BV antibracket in the local BRST cohomology~\cite{Barnich:1993pa,Henneaux:1997bm,Barnich:2000me}, see also \cite{Henneaux:1997ha,Barnich:2001mc,Bekaert:2002uh,Boulanger:2006gr} for applications. Furthermore, in this approach the deformations theory of local gauge systems fits the standard framework~\cite{Gerstenhaber:1964,Kontsevich:1997vb} for deformations of algebraic structures, based on so-called differential graded Lie algebras (DGLA), for a review see e.g.~\cite{Manetti2005}. 

Despite being very elegant and powerful, the local BRST cohomology approach to consistent deformations can be quite involved technically because it requires working in terms of infinite-dimensional jet-bundles and 
quotient spaces such as the space of local functionals. At the same time, in the case of topological models the procedure can be reformulated in terms of the finite-dimensional manifolds using the so-called AKSZ construction~\cite{Alexandrov:1995kv}. The idea to employ AKSZ-like formalism to study consistent deformations appeared in~\cite{Barnich:2009jy} and then was successfully applied to higher-spin interactions in 3 dimensions in~\cite{Grigoriev:2019xmp,Grigoriev:2020lzu}. The main advantage of the approach is that it is immediately on-shell,  locality is taken care of by the formalism itself, and the underlying space is (effectively) finite-dimensional. However, in this framework  one does not generally get a full control over the Lagrangian formulation at all orders and the approach is somewhat limited to topological systems. 

As far as general gauge theories are concerned an interesting extension of the AKSZ construction is the so-called presymplectic BV-AKSZ formulation which has been recently shown to encode a local gauge theory in terms of a finite-dimensional graded geometry~\cite{Grigoriev:2022zlq,Dneprov:2024cvt}, see also~\cite{Alkalaev:2013hta,Grigoriev:2016wmk,Grigoriev:2020xec,Dneprov:2022jyn} for earlier but less general versions.\footnote{It is also worth mentioning the approach of \cite{Barnich:2010sw,Grigoriev:2010ic,Grigoriev:2019ojp} 
which allows one to represent a given gauge theory as an infnite-dimensional AKSZ sigma model and which was a precursor of the presymplectic BV-AKSZ framework, see also~\cite{Stora:1984,Barnich:1995ap,Brandt:1996mh} for a relevant earlier developments in the context of local BRST cohomology. Another related earlier development is the so-called unfolded formulation~\cite{Vasiliev:1988xc,Vasiliev:2005zu}, where the equations of motion of a gauge system take the form of a free differential algebra, see~\cite{Barnich:2006hbb} for more details on the relation to the AKSZ sigma-models. For alternative approaches to extend AKSZ construction to nontopological models see e.g.~\cite{costello2011renormalization,Bonechi:2022aji}.} Presymplectic BV-AKSZ formalism can be considered as a far-going generalization of the AKSZ construction to the case of not necessarily topological and diffeomorphism-invariant theories. 

In this work we develop the procedure to analyse consistent interactions within the presymplectic BV-AKSZ approach. We restrict ourselves to the situation where the presymplectic structure stays intact under the deformations and hence restrict ourselves to strictly local interactions.  The important difference with the standard deformation theory is that in our approach two structures are deformed simultaneously -- the "covariant Hamiltonian"  and the $Q$ differential. However, it turns out that this does not lead to complications because the analysis breaks down into two consecutive steps: finding an infinitesimal deformation of the Hamiltonian, which is controlled by the standard homological setup, and then constructing the associated deformation of the $Q$ structure. The difference, however, is that the second step could also obstruct the deformation, in general.  

The paper is organized as follows. Section~\bref{sec:prelim} contains background material on the deformation theory of algebraic structures, Batalin-Vilkovisky formulation and its jet-bundle extension, and the presymplectic BV-AKSZ approach. In the main Section~\bref{sect:main} we develop a deformation theory in the presymplectic BV-AKSZ formalism and spell out explicitly the recursive procedure for studying interactions. In Section~\bref{sec:applications} we illustrate the approach using the examples of Chern-Simons and Yang-Mills theories.

\section{Preliminaries}
\label{sec:prelim}
\subsection{Deformation theory}
\label{sec:deform-theory}

A systematic framework~\cite{Gerstenhaber:1964} to study the deformations of algebraic structures is provided by differential graded Lie algebras (DGLA). In so doing, the infinitesimal deformations and obstructions acquire a homological interpretation. Now we briefly recall how it works without trying to be maximally general. For a systematic exposition and further references, see e.g.~\cite{Manetti2005}. 

By definition, a differential graded Lie algebra is a $\fZ$-graded linear space $L$ equipped with a differential $\delta:L\to L$ satisfying $\deg{\delta}=1$ , $\delta^2=0$ and the compatible graded Lie bracket $\commut{\cdot}{\cdot}$ of degree $0$. In applications, the differential $\delta$ encodes the underlying algebraic object. 
Typically, one is interested in modifying a given ``undeformed'' differential $\delta=\delta_0$ so that in the DGLA terms the bracket is unchanged while the differential gets deformed to $\delta=\delta_0+\commut{M}{\cdot}$ such that $(\delta,\commut{}{})$ is still a DGLA. 
%\footnote{\jordi{Maybe introduce the deformed $\delta$ and the MC element $M$ after the MC equation \eqref{eq_mc_equation}?}}
Condition $\delta^2=0$ implies that the element $\delta_0 M+\half\commut{M}{M}$ belongs to the centre, i.e. its adjoint action vanishes. Assuming that the degree $2$ component of the centre is trivial, one arrives at the Maurer-Cartan (MC) equation:
\begin{equation}
\label{eq_mc_equation}
\delta_0 M+\half\commut{M}{M}=0\,.
\end{equation}

It is often convenient to consider formal deformations. This is achieved by tensoring $L$ with formal series $\fC[[g]]$ in the deformation parameter $g$. Assuming $\delta_0$ and $\commut{}{}$ to be of order $0$ in $g$ and 
\begin{equation}
M=\sum_{i=1}^\infty M_i    \,,
\end{equation}
where $M_i$ is of order $i$ in $g$,
the MC equation takes the form
\begin{equation}
\label{MC}
\delta_0 M_k
+\hhalf\sum_{l=1}^{k}
\commut{M_l}{M_{k-l}}=0\,, \quad k \geq 1
\end{equation}
For instance, for $k=1$ one gets $\delta_0 M_1=0$ so that the first order deformations are determined by $\delta_0$-cocycles in degree $1$.

Not every first-order deformation can be extended to the next order. Indeed,
at order $2$ equation~\eqref{MC} says that $\delta M_2+\half\commut{M_1}{M_1}=0$ or in other words $\commut{M_1}{M_1}$ must be trivial in the cohomology of $\delta_0$. If it is not, one says that the deformation is obstructed.

It is clear that at order $k$ one can always add a transformation of the form  $M_k \to M_k+\delta_0 T_k$. Such deformations are unobstructed and correspond to the automorphisms of $L$ of the form $\exp{\commut{T_k}{\cdot}}$ and hence are naturally considered trivial.

To illustrate the general framework let us consider the Lie algebra deformations. Given a linear space $\algg$ the Lie algebra structure on $\algg$ can be encoded in the homological vector field $d_{CE}$ on $\algg[1]$. In other words cochains of $\algg$ are identified with functions on $\algg[1]$ and the Chevalley-Eilenberg differential of $\algg$ is identified with $d_{CE}$. The associated DGLA is that of vector fields on $\algg[1]$ with the degree being the homogeneity degree and the graded Lie bracket being the commutator. The undeformed $d_{CE}$ gives rise to the differential $\delta_0:L\to L$, $\delta_0 a=\commut{d_{CE}}{a}$. In this case the MC equation can be written as:
\begin{equation}
\commut{d_{CE}+M}{d_{CE}+M}=0\,.
\end{equation}
In a more general setup, one could allow $\delta_0$ to be an outer derivation of $\commut{\cdot}{\cdot}$.

Let us also give a coordinate expression for the structures involved in this example. Let $e_i$ be a basis in $\algg$
and the Lie algebra operation on $\algg$ to be determined by the structure constants $k_{ij}{}^k$ defined through $\commut{e_i}{e_j}_{\algg}=k_{ij}{}^k e_k$. Coordinated on $\algg[1]$ in the basis $e_i$ are $c^i$, $\deg{c^i=1}$. The Chevalley-Eilenberg differential $d_{CE}$ is given by:
\begin{equation}
    \label{ce_differential}
    d_{CE}=-\hhalf c^ic^j k_{ij}{}^k\dl{c^k}\,.
\end{equation}
The nilpotency condition $(d_{CE})^2=0$ is equivalent to the Jacobi identity of $\algg$. A generic element of $L$ of degree $1$ is given by $M^\prime=d_{CE}-\half c^ic^j M_{ij}{}^k\dl{c^k}$ and the MC equation implies that $k_{ij}{}^k+M_{ij}{}^k$ satisfy Jacobi identity and hence define a new Lie algebra structure on $\algg$. If $\algg$ is  a graded space, the generic element $M^\prime$ of total degree $1$ is not necessarily linear in $c^i$ and in fact defines an $L_{\infty}$ algebra structure on $\algg$, provided that $M^\prime$ satisfies the Maurer-Cartan equation. In other words, this is a deformation in the $L_{\infty}$ sense. This example illustrates how the deformation approach may also help to identify more general algebraic structures.

\subsection{BV formulation of gauge theories}

It was observed in \cite{Barnich:1993vg} that the introduction of interactions can be reformulated as a deformation of the underlying DGLA which, it turn, is determined by the Batalin-Vilkovisky formulation of the gauge theory in question. 

In the BV formalism a gauge theory is described in terms of the field-antifield spaceб which is a $\fZ$-graded supermanifold graded by the ghost degree $\gh{\cdot}$. This space is equipped with an odd Poisson bracket (BV antibracket) which carries ghost degree $1$ and satisfies the graded analogues of the Leibniz and Jacobi identities. The information about the action and (higher) gauge transformation and their algebra is encoded in the unique function $S_{BV}$, $\gh{S_{BV}}=0$, called the \textit{master action}. At the same time, the gauge invariance of the action and further compatibility relations are encoded in the so-called master equation:
\begin{equation}
\ab{S_{BV}}{S_{BV}}=0\,,
\end{equation}
where $\ab{}{}$ denotes the odd Poisson bracket.

The space of functions (strictly speaking, space of functionals) on the field-antifield space can be seen as a DGLA if one takes $\delta=\ab{S_{BV}}{\cdot}$ and shifts the ghost degree by $1$ so that after the shif the degree of the antibracket becomes zero. In such a framework  one can study deformations of the gauge theories either by introducing deformation parameter $g$ or simply using homogeneity in fields. In the latter case, one assumes that the initial theory is free, i.e. $S_{BV}$ is quadratic in fields. This precisely corresponds to the standard QFT approach where one starts with free fields and then introduces interaction vertexes.

Although BV formulation defines a standard deformation theory, at least at the algebraic level, the situation is more involved because the field-antifield space is typically infinite-dimensional and the result crucially depends on the choice of the class of functionals which constitute the linear space underlying the respective DGLA.

The usual and  physically motivated choice is given by local functionals, i.e. functionals whose integrands depend on only finite number of derivatives of the fields. In the standard framework, this space can be described as the quotient space of local functions modulo total derivatives so that the antibracket descends to the quotient, resulting in a DGLA structure therein~\cite{Barnich:1997ed}. This is of course a refinement of the standard QFT framework where the interaction vertexes are defined modulo total derivatives.

\subsection{BV for local gauge theories and consistent deformations}

The standard mathematical framework to study local gauge theories is that of jet-bundles. This has a natural extension  to the case where the fibres are $\fZ$-graded supermanifolds, which is crucial in the BV context. Let $\cE\to X$ be a bundle over a real space-time manifold $X$. Consider the algebra $\bigwedge^\bullet(J^\infty(\mathcal{E}))$ of local differential forms on $J^\infty(\mathcal{E})$. Local forms are those that depend on finite number of coordinates, i.e. those that can be represented as the pullbacks of forms from finite jets. Because infinite jet bundle has a canonical horizontal distribution (Cartan distribution),  $\bigwedge^\bullet(J^\infty(\mathcal{E}))$ is actually bigraded and the total form degree is a sum of vertical and horizontal form degree. An $r$-horizontal and $s$-vertical form on $J^\infty(\mathcal{E})$ is refereed to as an $(r,s)$-form. The de Rham differential on $J^\infty(\mathcal{E})$ decomposes into the vertical and horizontal parts:
\begin{align}
\begin{gathered}
	d=\dh+\dv\,, \qquad \dh\coloneqq dx^a D_a, \\
    \dh^2=0,\qquad \dv^2=0, \qquad [\dh,\dv] = 0.
\end{gathered}
\end{align}
In the previous expression, $D_a$ is the total derivative $D_a = \pdv{x^a} + \Phi^A_{Ia}\pdv{\Phi^A_I}$, with $x^a$ the spacetime coordinates, $\Phi^A$ the field coordinates, and $\Phi^A_I$ all its derivatives (in the jet bundle sense). We denote by \textit{$(r,s)$-form} a $r$-horizontal and $s$-vertical form on $J^\infty(\mathcal{E})$.
More details about the geometry of jet spaces can be found in e.g. \cite{Anderson1991,Krasil'shchik:2010ij}. 

The locality of the BV formalism can be made manifest by reformulating it in the language of jet-bundles~\cite{Henneaux:1993hv,Barnich:1995db}, see also~\cite{Piguet:1995er,Barnich:2000zw}. More specifically, a \textit{local BV system} $(\mathcal{E},s,\omega)$ is a graded  fibre bundle $\mathcal{E}\rightarrow X$ over a real spacetime manifold $X$ together with a nilpotent and evolutionary vector field $s$, $\gh{s}=1$ defined on $J^\infty(\mathcal{E})$. In addition, $\cE$ is equipped with a closed  $n+2$ form $\omega_\cE$ of ghost degree $-1$ such that it does not have zero vectors and vanishes on any $3$-vertical vectors. We denote by $\omega$ its pullback to $J^\infty(\mathcal{E})$ and require $s$ and $\omega$ to be compatible in the sense that $L_s\omega=\dh(\cdot)$, i.e. $L_s\omega$ is trivial in $\dh$-cohomology. We refer to $s$ and $\omega$ as to \textit{BRST differential} and the \textit{BV symplectic structure} respectively. Note that $\dv\omega=\dh\omega=0$ and that $\omega$ is necessarily an $(n,2)$-form.

Given a local BV system one finds that $L_s\omega=\dh(\cdot)$ implies that (at least locally) there exists an $(n,0)$-form $\cL_{BV}$ of ghost degree $0$ such that $\dv \cL_{BV}+ i_s \omega=\dh \Theta$  for some $(n-1,1)$-form $\Theta$. $\cL_{BV}$ is referred to as \textit{BV Lagrangian} and defines a BV action functional $S_{BV}$ via
\begin{equation}
S_{BV}[\sigma]=\int_X \bar\sigma^*(\cL_{BV})\,,
\end{equation}
where $\sigma$ is a supersection of $\cE$ and $\bar\sigma: X \to J^\infty(\cE)$ its jet-prolongation.

Consider an $(n,0)$-form $f$ and an evolutionary vector field $V_f$ such that $\iota_{V_f}\omega +\dv f = \dh(\dots)$. We say that $V_f$ is a \textit{Hamiltonian vector field} with \textit{Hamiltonian} $f$. For any such $f$, its associated Hamiltonian vector field exists and is unique because $\omega^\cE$ does not have zero vectors. The other way around, given an evolutionary $V$ satisfying $L_V\omega=\dh(\cdot)$, there exists (at least locally)  $H_V$ satisfying $\iota_{V}\omega +\dv H_V = \dh(\dots)$. In particular, the BRST differential $s$ defines the \textit{BV Lagrangian} $\cL_{BV}$.

Let $V_f$ and $V_g$ be Hamiltonian vector fields of $f$ and $g$ respectively. Consider the following binary operation:
\begin{equation}
	\ab{f}{g} \coloneqq \iota_{V_f}\iota_{V_g}\omega\,.
\end{equation}
It is well-defined on the equivalence classes of $(n,0)$-forms modulo $d_h$-exact forms and moreover defines a degree $1$
graded Lie algebra structure therein (by shifting the degree it becomes a standard graded Lie algebra). This quotient space is known as the space of \textit{local functionals} and can be identified with $H^{(n,0)}(\dh)$.

One can check that the BV Lagrangian $\cL_{BV}$ defined through $\dv \cL_{BV}+\iota_s\omega=\dh(\cdot)$ satisfies
\begin{equation}
\ab{\cL_{BV}}{\cL_{BV}} = \dh(\cdot)\,,
\end{equation}
thanks to $s^2=0$. This is referred to as the \textit{classical master equation}. Conversely, given $\cL_{BV}$, $\gh{\cL_{BV}}=0$ satisfying the master equation, its Hamiltonian vector field is nilpotent. In this way, a local BV system can be defined by giving $\cL_{BV},\omega$ instead of $s,\omega$. See e.g. \cite{Sharapov:2016sgx,Grigoriev:2022zlq} for further details and proofs.

To summarise, a local BV system gives rise to a (shifted) DGLA where the $\delta$-differential and the Lie bracket are given by $s$ and $\ab{\cdot}{\cdot}$ respectively, and both are understood acting on local functionals. The DGLA deformation theory recalled in Section~\bref{sec:deform-theory}  can now be applied. The resulting deformation theory is of course just the one developed in~\cite{Barnich:1993vg,Barnich:1993pa} from the gauge theoretical perspective.

The triple $(H^{n,0}(d_h),s,(\cdot,\cdot))$ is a shifted DGLA.  The Maurer-Cartan elements of this DGLA are equivalence classes $[T]$ of ghost degree $0$ such that $s[T] + \frac{1}{2}([T],[T]) = 0$. They define the deformations of the BV Lagrangian (master action) 
\begin{equation}
	(s+([T],\cdot))^2 = 0 \Leftrightarrow s[T] + \frac{1}{2}([T],[T]) = 0.
\end{equation}
Thus, $s' = s+([T],\cdot)$ is a deformation of the BRST differential $s$. Its associated BV Lagrangian is just $\cL_{BV}+T$. Recall that we adapt the usual assumption  that the symplectic structure is intact under the deformation.  

\iffalse
\begin{equation}
	L = L_0 + gL_1 + g^2L_2 + \dots
\end{equation}
of the BV Lagrangian $L_0$. This modification represents a deformation of the master equation if $([L],[L])=0$. Decomposing it by order in $g$, we get the equations
\begin{align}
	([L_0],[L_1]) = 0, \quad ([L_0],[L_2]) = -\frac{1}{2}([L_1],[L_1]), \dots&,\\
    ([L_0],[L_i]) = -\frac{1}{2}\sum_{k=1}^{i-1}([L_k],[L_{i-k}]), \dots&
\end{align}
This set of equations can be systematically analyzed only when knowing the cohomology of the BRST differential $q_0 = ([L_0],\cdot)$. If it was indeed possible to solve them, then we say that the vector field $q = ([L_0 + gL_1 + g^2L_2 + \dots],\cdot)$ is a consistent deformation of the BRST differential.

Having solved the previous equations for representatives $L_0,\dots,L_{i}$, if $\sum_{k=1}^{i}([L_k],[L_{i+1-k}])$ is not $q_0$-exact, we say that it represents an \textit{obstruction} and imposes conditions on $L_{j\leq i}$ in order to satisfy $([L],[L])=0$. In addition, a representative of the transformation $[L'] = [L] + q_0 [T]$ for some $T$ is said to be an \textit{infinitesimal trivial deformation} and corresponds to a field redefinition in the space of functions over $X$.

\fi
\subsection{Presymplectic BV-AKSZ formulation}

The geometrical setup for various versions of the BV-AKSZ approach is provided by graded geometry and, more specifically, fibre bundles where not only fibres but also base spaces are graded manifolds. Given a fibre bundle $\pi: E\to \cX$, let $\cI$ be the ideal in $ \bigwedge^\bullet(E) $ generated by the elements of the form $\pi^*\alpha$, with $\alpha \in \bigwedge^{k>0}(\cX)$. The quotient algebra $ \bigwedge^\bullet(E) /\cI$ can be identified with vertical forms on $E$. Note that the de Rham differential preserves $\cI$ and hence descends to the quotient. The same applies to $L_V$ with $V$ a projectable vector
field on $E$. Recall that $V$ is projectable if $ V \circ \pi^* = \pi^* \circ v $ for a vector field $v$ on $\cX$. Although interior product $\iota_V$ does not generally preserve $\cI$,  it does so if $V$ is vertical.

Let $X$ be a real spacetime manifold. Differential forms on $X$ can be thought of as functions on the graded manifold $T[1]X$. It comes equipped with a canonical vector field  $\dx$ which is just the de Rham differential of $X$ seen as a vector field on $T[1]X$. If $x^a$ are local coordinates on $X$ the associated coordinates on $T[1]X$ are $ \{x^a,\theta^a\} $, $\gh{\theta^a}=1$. In these coordinates one has $\dx=\theta^a\dl{x^a}$. In what follows we limit ourselves to fibre bundles over $T[1]X$, where $X$ is a real manifold which serves as a space-time manifold.  

The central object of our study is the presymplectic gPDE\footnote{In \cite{Grigoriev:2022zlq,Dneprov:2024cvt} this object was called weak presymplectic gauge PDE. Here we drop "weak" for shortness. Genuine, i.e. strict in contrast to weak, presymplectic gPDEs are not considered in this work.}:
\begin{definition}
	A \textit{presymplectic gauge PDE} $(E,Q,\omega,\cL,T[1]X)$ is a graded fibre bundle $\pi:E\rightarrow T[1]X$ equipped with a
    vector field $Q$, $\gh{Q}=1$, presymplectic structure $\omega$, $\gh{\omega}=\dim X-1$, and a function $\cL$, $\gh{\cL}=n$ such that:
	\begin{gather}
		Q\circ\pi^* = \pi^*\circ \dx, \qquad d\omega = 0, \\ 
        \label{iqomega}
        i_Q\omega+d\cL \in \cI, \qquad Q\cL+\hhalf \iota_Q\iota_Q \omega = 0\,. 
	\end{gather}
\end{definition}
Because we always denote spacetime by $X$, a shortened notation $(E,Q,\omega,\cL)$ is often employed below.  The 
data of presymplectic gPDE defines an action functional on the space of sections $T[1]X\to E$:
\begin{equation}
\label{pgpde-action}
S[\sigma]=\int_{T[1]X}\sigma^*(\chi)(\dx)+\sigma^*(\cL)\,,
\end{equation}
where $\chi$, $\omega=d\chi$ is the presymplectic potential.  Note that the ambiguity in $\chi \to \chi +d\lambda$ results in adding a total-derivative to the integrand of \eqref{pgpde-action}.
The gauge transformations of this action are determined by $Q$. Moreover,
the entire BV formulation of the underlying gauge theory is encoded in 
the data of presymplectic gPDE, see~\cite{Grigoriev:2022zlq,Dneprov:2024cvt} for further details and \cite{Alkalaev:2013hta,Grigoriev:2016wmk,Grigoriev:2020xec,Dneprov:2022jyn} for earlier but less general constructions. 

The following comments are in order: the symplectic structure defines
 the action only modulo boundary terms. Indeed, passing from $\chi$ to $\chi+d\alpha$ changes the integrand by a total derivative ($\dx$-exact term). In this sense it is probably preferable to define the theory by giving $(E,Q,\chi,\cL)$. Moreover, $\cL$ is determined by $Q,\chi$ up to field-independent terms. Although $\omega$ is typically not regular, it is required to be quasi-regular so that it determines a genuine BV symplectic structure on the symplectic quotient of the space of super-sections, see~\cite{Dneprov:2024cvt,Grigoriev:2024ncm} for more details.

The presymplectic structure on $E$ determines a suitable version of the kernel distribution. Namely, we take $\wker$ to be a distribution on $E$ generated by vertical vector fields $V$ satisfying $\iota_V\omega \in \cI$. By a distribution we mean a submodule of a module of vector fields so it is not necessarily regular (constant rank). Note that $Q^2 \in \wker$. Indeed, 
\begin{align}
	2\iota_{Q^2}\omega &= \iota_{[Q,Q]}\omega = [\mathcal{L}_Q,\iota_Q]\omega = 2\iota_Qd \alpha - d \iota_Q\alpha + dQ\h \\
	&= 2\iota_Qd \alpha - d \iota_Q\alpha + d(-\iota_Q\alpha) = 2\mathcal{L}_Q\alpha \in \mathcal{I},
\end{align}
where $ \iota_Q\omega + d\h = \alpha \in \mathcal{I} $, and we used that $ Q\h = -\iota_Q\alpha $:
\begin{equation}
	0 = Q\h + \frac{1}{2}\iota_Q\iota_Q\omega = Q\h + \frac{1}{2}\iota_Q(-d\h + \alpha)  = \frac{1}{2}(Q\h + \iota_Q\alpha).
\end{equation}
Furthermore, adding a vertical vector field $V\in \wker$ to $Q$ does not affect the defining conditions of a presymplectic gPDE. Indeed, for $Q^\prime=Q+V$ we have:
\begin{align}
	\iota_{Q'}\omega + d\h = \iota_{Q}\omega + d\h, \quad Q'\h + \frac{1}{2}\iota_{Q'}\iota_{Q'}\omega = Q\h + \frac{1}{2}\iota_Q\iota_Q\omega\,.
\end{align}
It follows, $Q$ can be regarded as an equivalence class of vector fields modulo $\wker$.

It is convenient to introduce a subalgebra $\algA\subset \cC^\infty(E)$ of functions annihilated by $\wker$. In particular, $Q$ preserves $\algA$ and is nilpotent there, i.e. $Q^2f=0$ for any $f \in \algA$. Furthermore, it makes sense to speak about $Q$-cohomology in $\algA$, or more generally $L_Q$ cohomology in the subalgebra of forms preserved by $L_V$, for all $V \in \wker$.

\section{Consistent deformations of presymplectic gPDEs}
\label{sect:main}

\subsection{General setup}
Although a presymplectic BV-AKSZ system does not directly determine a DGLA that can be used in studying consistent deformations\footnote{In principle, a presymplectic gPDE  determines \cite{Grigoriev:2022zlq,Dneprov:2024cvt} a local BV system which in turn defines a DGLA, but this is not what we are after now.}, it turns out that one can still develop an appropriate  analogue of the deformation theory.

Let $(E,Q_0,\omega,\h_0)$ be a linear presymplectic gPDE. We are interested in its consistent deformations. We restrict ourselves to the deformations that do not affect the presymplectic structure $\omega$. Such deformations preserve the derivative order of the system and the number of degrees of freedom. Of course, this is not the most general class of deformations but it is well-defined and is in a certain sense natural. Note that even if we deform the presymplectic structure we can always undo the deformation by a suitable coordinate change provided that the structure satisfies certain regularity conditions and its rank is not affected by the deformation.

Let $\deg_\psi$ denote the homogeneity degree in fields (fibre coordinates).
Because our starting point system is linear we have $\deg_\psi (Q_0) = 0$, $\deg_\psi (\h_0) = 2$ and $\deg_\psi(\omega) = 2$. We refer to the homogeneity degree in fields as to \textit{order} from now on. We intend to complete $Q_0$ and $\h_0$ to
\begin{align}
	Q = Q_0 + Q_1 + \dots,\quad \h = \h_0 + \h_1 + \dots\,,
\end{align}
where
\begin{equation}
    \deg_\psi \h_i = \deg_\psi Q_i + 2 = i + 2\,,
\end{equation}
in such a way that $(E,Q,\omega,\cL)$ is still a presymplectic gPDE. In particular this means that $Q$ should project to $\dx$, i.e. $\pi^*\circ \dx = Q\circ \pi^*$ and hence all $Q_{i}$, $i \geq 1$ should be vertical, since $Q_0$ projects to $\dx$.

It is useful to introduce the following conventions for the partial sums
\begin{align}
	Q_{(i)} \coloneqq Q_0 + \dots + Q_i, \qquad \h_{(i)} \coloneqq \h_0 + \dots + \h_i
\end{align}
and denote \textit{the order-$k$ component of $X$} by $\left.X\right|_k$. Furthermore,   $X = \order{k}$ means that \textit{$X$ is of order $k$ or higher}. In other words, if $X=\order{k}$, then $X|_{\leq k-1} = 0$. 

We have the following:
\begin{proposition}\label{prop:main-rec}
Given a linear pgPDE $(E,Q_0,\omega,\h_0)$, consider its deformation  $(E,Q_{(N)},\omega,\h_{(N)})$ such that $Q_i, \h_i$ fulfill
	\begin{gather}\label{eq_build_def_1_2}
		\iota_{Q_i}\omega + d\h_i  \in \mathcal{I},  \\
        \label{eq_build_def_1_1}
		Q_0\h_i = -\frac{1}{2}\sum_{k=1}^{i-1}Q_k \h_{i-k}
	\end{gather}
	for all $1\leq i\leq N$. Then the defining relations of  pgPDE are satisfied up to order $N+3$:
	\begin{align}
    \label{1st-recursive}
		\iota_{Q_{(N)}}\omega + d\h_{(N)} \in \cI, \qquad Q_{(N)}\h_{(N)} + \frac{1}{2}\iota_{Q_{(N)}}\iota_{Q_{(N)}} \omega = \order{N+3}\,.
	\end{align}
    \end{proposition}
\begin{proof} The proof is done by analysing the homogeneous components of the defining relations. The only nontrivial point is that in obtaining the second relation one employs $\iota_{Q_l}\iota_{Q_m}\omega=-Q_l\cL_m$, $l\geq 1$,  which holds because
    $\iota_{Q_m}\omega=-d\cL_m+\cI$ and $Q_l$ is vertical for $l\geq 1$.
\end{proof}
Note that $\cL_i$, belong to $\algA$ for $i\geq 1$. Indeed, applying $\iota_V$, $V\in \wker$ to the first equation of \eqref{1st-recursive} one finds $V\cL_i=-\iota_{Q_i}\iota_V\omega=0$ because $Q_i$ is vertical for $i\geq 1$. Another important property of the deformation $(E,Q_{(N)},\omega,\h_{(N)})$ is that it satisfies the consistency condition. Namely,
\begin{proposition}\label{prop_consistency}
A collection $\{ Q_0,\dots,Q_{N}; \h_0,\dots,\h_{N} \}$ such that \eqref{eq_build_def_1_2} and  \eqref{eq_build_def_1_1} are fulfilled for all $1\leq j\leq N$ satisfies the consistency condition:
	\begin{align}
		\label{eq_consistency_1}
		Q_0\left( \sum_{k=1}^{N} Q_k \h_{N+1-k} \right) = 0\,.
	\end{align}
    \end{proposition}
\begin{proof}
Let $q$, $\gh{q}=1$, be an odd vector field satisfying $\iota_q \omega+d l\in\cI$ and such that $q^2$ is vertical. The following identity holds:
\begin{equation}
\label{mcons}
q(\hhalf \iota_q\iota_q\omega+ql)=0\,.
\end{equation}
Indeed, the first term gives
\begin{multline}
\label{1st-aux}
\hhalf L_q \iota_q\iota_q \omega=2\iota_{q^2}\iota_q\omega+\hhalf \iota_q\iota_q  L_q  \omega= \\
2\iota_{q^2}(-dl+\cI)+\hhalf \iota_q\iota_q  L_q  \omega= -2q^2l +\hhalf \iota_q\iota_q  L_q  \omega\,,
\end{multline}
where we used that $q^2$ is vertical. Furthermore,
\begin{multline}
d(\iota_q\iota_qL_q\omega)=
-\iota_qL_qL_q\omega-L_q\iota_qL_q\omega=
\\
-\iota_{\commut{q}{q}}L_Q\omega-2\iota_qL_{q^2}\omega=-2 L_{q^2}\iota_q \omega =
2d (q^2 l)+\cI \,.
\end{multline}
It follows $d(\hhalf \iota_q\iota_qL_q\omega-q^2l)\in \cI$ and hence has the form $\pi^*g$ for some $g\in \cC^\infty(T[1]X)$. It follows, $g=0$ because  there are no nontrival functions of ghost degree $n+1$ on $T[1]X$. This shows that $\hhalf \iota_q\iota_qL_q\omega=q^2l$ which together with \eqref{1st-aux} gives~\eqref{mcons}.

Let us then take $q=Q_{(N)}=Q_0+\ldots+Q_N$ and
$l=\cL_{(N)}=\cL_0+\ldots+\cL_N$ satisfying~\eqref{eq_build_def_1_2} and \eqref{eq_build_def_1_1}. Using \eqref{1st-recursive} one finds
\begin{multline}
\left(Q_{(N)}\left(\frac{1}{2}\iota_{Q_{(N)}}\iota_{Q_{(N)}} \omega +Q_{(N)}\cL_{(N)}\right)\right)\Big|_{N+3}=\\
Q_0\left(\left(\frac{1}{2}\iota_{Q_{(N)}}\iota_{Q_{(N)}} \omega+Q_{(N)}\cL_{(N)}\right)\Big|_{N+3}\right)=
\hhalf Q_0 \left( \sum_{k=1}^{N} Q_k \h_{N+1-k} \right)
\end{multline}
where in the last equality we used \eqref{eq_build_def_1_1} and the fact that $Q_k$ is vertical for $k>0$.
\end{proof}

The above Propositions determine a recursive procedure for constructing consistent deformations.  Indeed, starting from a linear presymplectic gPDE $(E,Q_0,\omega,\h_0)$ one can attempt constructing   $Q_i,\cL_i$ by solving \eqref{eq_build_def_1_1} and \eqref{eq_build_def_1_2} order by order. Unless one reaches an obstruction, i.e. at some order the cocycle in~\eqref{eq_consistency_1} is nontrivial in $Q_0$-cohomology or equation \eqref{eq_build_def_1_2} does not have a solution, the procedure either stops at certain order, giving a presymplectic gPDE with vertices of finite order, or produces an infinite formal deformation. In the later case one usually proves that the deformation is unobstructed at higher orders.

Conversely, given a presymplectic gPDE $(E,Q,\omega,\h)$ and its particular solution $\sigma$, the expansion of $Q,\cL$ around $\sigma$ gives a consistent deformation of the linearization of $(E,Q,\omega,\h)$ around $\sigma$. Details on how the perturbative expansion of (presymplectic) gPDEs works can be found in~\cite{Grigoriev:2024ncm}.

%%%%%%%%%%%%%%%%%%%%%%%%%%%%%%%%%%%%%%%%
%%%%%%%%%%%%%%%%%%%%%%%%%%%%%%%%%%%%%%%%%
\subsection{Recursive  construction }
Now we explicitly spell-out the recursive procedure for analysing consistent deformations. Given a linear presymplectic gPDE $(E,Q_0,\omega,\h_0)$, equations \eqref{eq_build_def_1_1} and \eqref{eq_build_def_1_2} for $i=1$ restrict the first order deformations. Namely, 
\begin{align}
	\iota_{Q_1}\omega + d \h_1 &\simeq 0 \label{eq_1stdefs_main1} \\
	Q_0\h_1 &= 0. \label{eq_1stdefs_main2}
\end{align}
where $\simeq$ denotes equality modulo $\cI$, i.e. $A\simeq 0$ means that $A\in \cI$.
The strategy is then to look for $\cL_1$ of ghost degree $n$, cubic in fields, and  satisfying \eqref{eq_1stdefs_main2}. The first condition forces $\cL_1$ to belong to $\algA$. Moreover, vector field $Q_0$ is nilpotent in $\algA$  and hence $\cL_1$ belongs to the cohomology $H^n(Q_0,\mathcal{A})$. As we are going to see later, if $\cL_1$ is a  representative of a trivial cohomology class, the deformation it determines is trivial.

The condition $\cL_1 \in \algA$ is necessary for \eqref{eq_1stdefs_main1} to have a solution for $Q_1$ but can be not sufficient in general. If $\omega$ is regular it is easy to check that $Q_1$ must exist, at least locally. However, in applications $\omega$ is often not regular but satisfies the more general condition of quasi-regularity, so that the existence of $Q_1$ is not guaranteed and might result in an additional obstruction. We say that $\cL_1 \in \algA$ such that $Q_0\cL_1=0$ and $d\cL_1 \in \im(\omega)$ represents an \textit{infinitesimal deformation} of our initial system. Inequivalent infinitesimal deformations constitute the tangent space to the moduli space of deformations.

Given an infinitesimal deformation $(\cL_1,Q_1)$ the analysis of higher orders is done order by order, just like in the case of DGLA. However, at each order the analysis is slightly different. More precisely, suppose we have constructed $\{ Q_0,\dots,Q_{N-1}; \h_0,\dots,\h_{N-1} \}$ such that \eqref{eq_build_def_1_1} and \eqref{eq_build_def_1_2} are satisfied for $0<i<N$. Then we are looking for $\h_N \in \algA$, such that  $ \gh{\h_N} = n $, $ \deg_\psi(\h_N) = N+2 $ and  Eqs.  \eqref{eq_build_def_1_1} and \eqref{eq_build_def_1_2} hold for $i=N$, i.e. 
	\begin{gather}
    \label{eq_method_assumptions_1}
		Q_0 \h_N+\frac{1}{2} B_N = 0\,, \qquad
        B_N \coloneq \sum_{k=1}^{N-1} Q_k \h_{N-k}\,,\\
		\iota_{Q_N} \omega + d \h_N \simeq 0\,.
        \label{eq_method_assumptions_2}
	\end{gather}

Thanks to the first equation, $\h_N \in \algA$, so that the necessary condition for the existence of a solution to \eqref{eq_method_assumptions_1} is that $B_N$ is $Q_0$-closed. We refer to the requirement $Q_0 B_N=0$ as the \textit{consistency condition}. It is satisfied thanks to the Proposition~\bref{prop_consistency}.

A sufficient condition for the existence of a solution to \eqref{eq_method_assumptions_1} is that $B_N$ is $Q_0$-exact. If the cohomology class of $B_N$ is nontrivial we say that the deformation is \textit{obstructed}. In other words, possible obstructions are non-trivial elements of $H^{n+1}(Q_0,\mathcal{A})$. If $B_N$ is nontrivial, $\{ Q_1,\dots,Q_{N-1}; \h_1,\dots,\h_{N-1} \}$ should satisfy some additional conditions in order to make $B_N$ trivial in $H^{n+1}(Q_0,\mathcal{A})$. If it is not possible to satisfy these conditions for a nontrivial $\cL_1$ one usually says that the initial system is {\textit{rigid}}, i.e. does not admit  deformations as a presymplectic gPDE.

Finally, having obtained $\cL_N \in \algA$ that solves \eqref{eq_method_assumptions_1} we turn to the remaining equation \eqref{eq_method_assumptions_2}. It is identical to the analogous problem at the first order, which was discussed above. 
If \eqref{eq_method_assumptions_2} admits a solution for $Q_N$ this completes the $N$-th step of the recursive construction.

It can §happen that the deformation already constructed is complete in the sense that $Q_{(N)},\cL_{(N)}$ satisfy~\eqref{iqomega} and hence defines a presymplectic gPDE. In this case, the constructed deformation is finite.  If this does not happen or one is interested in studying further deformations one repeats the procedure at order $N+1$ and so on. This completes the detailed description of the recursive procedure.

In the next section we will detail the first steps while showing that \eqref{eq_consistency_1} is fulfilled at first orders.

%%%%%%%%%%%%%%%%%%%%%%%%%%%%%
%%%%%%%%%%%%%%%%%%%%%%%%%%%%%
\subsection{Lowest orders in more details}

The first order deformations have been explicitly considered above so that  we start with the second order. Let us spell-out explicitly  the derivation of  \eqref{eq_consistency_1} at order $2$. We have 
\begin{multline}
	Q_0B_1=Q_0( Q_1 L_1 ) = \mathcal{L}_{Q_0}Q_1\h_1 =  \mathcal{L}_{Q_0}\iota_{Q_1}d\h_1 = \mathcal{L}_{Q_0}\iota_{Q_1}(-\iota_{Q_1}\omega + \alpha_1) \\
	= -\mathcal{L}_{Q_0}\iota_{Q_1}\iota_{Q_1}\omega = -\iota_{[Q_0,Q_1]}\iota_{Q_1}\omega - \iota_{Q_1}\mathcal{L}_{Q_0}\iota_{Q_1}\omega \\
	 = - \iota_{Q_1}\mathcal{L}_{Q_0}\iota_{Q_1}\omega = -\iota_{Q_1}\mathcal{L}_{Q_0}(-d\h_1 + \alpha_1) = -\iota_{Q_1}\mathcal{L}_{Q_0}\alpha_1 = 0\,,
\end{multline}			
where we used $Q_0\cL_1=0$ and $\iota_{[Q_1,Q_0]}\omega\in \mathcal{I}$,
which holds thanks to:
\begin{align}
	\iota_{[Q_1,Q_0]}\omega = \mathcal{L}_{Q_0}\iota_{Q_1}\omega = Q_0\h_1 + \mathcal{L}_{Q_0}\alpha_1 = \mathcal{L}_{Q_0}\alpha_1 \in \mathcal{I} \label{eq_i10_omega}\,.
\end{align}
If we don't step into an obstruction, we can already solve
	\begin{equation}
		Q_0 \h_2 =  -\frac{1}{2} Q_1 \h_1
	\end{equation}
for $ \h_2 \in \algA$ with $ \deg_\psi(\h_2) = 4$, $\gh{\cL_2}=n$.
	
We then proceed to obtain a vertical $Q_2$, $ \deg_\psi(Q_2) = 2$ from
	\begin{equation}
		\iota_{Q_2} \omega + d \h_2 \simeq 0\,.
	\end{equation}
Again, one should note that $Q_2$ is defined up to a vector field from $\wker$.

In the case where $Q_1\h_1$ represents an obstruction, i.e., it is non trivial in $Q_0$-cohomology, one usually interprets this as a condition restricting the first order deformation. If it was not possible to make $Q_1\cL_1$ trivial for a nontrivial $\cL_1$ one concludes that the initial system $(E,Q_0,\omega,\h_0)$ was rigid.

If the obstruction vanishes we check whether $Q_{(2)}$ and $\h_{(2)}$ satisfy the conditions to all orders, i.e. whether $(E,Q_{(2)},\omega, \cL_{(2)})$ is a presymplectic gPDE. If it is not, we continue to the next order.

\subsection{Trivial deformations}
\label{sec_triv_def}

Just like in the case of DGLA, among possible deformations of a presymplectic gPDE there are some that do not actually produce a new system as they can always be undone by a field redefinition. 

\iffalse
\begin{definition}%[infinitesimal trival deformation]
	Given a presymplectic gPDE $(E,Q,\omega,\h)$, an \textit{infinitesimal trival deformation}
	is a transformation $\h'_N = \h_N + Q_0 R$ with $\deg_\psi(R) = N+2$, $\gh(R) = n-1$ and such that $dR$ is in the image of $\omega$ modulo $\mathcal{I}$.
\end{definition}
\fi 
\begin{definition}
	Given a presymplectic gPDE $(E,Q,\omega,\h)$, a \textit{trivial formal deformation}
	is defined as: 
	\begin{align}
		 \cL^\prime &= \exp(V)\cL+\iota_{\dx}\beta,\quad  \label{eq_def_trivb_1} \\
	 Q^\prime &= \exp(V) Q \exp(-V), \label{eq_def_trivb_2}
	\end{align}	
    where the exponential is understood as a formal series and $V$ is a vertical vector field satisfying $\deg_\psi(V) > 0$, $\gh{V} = 0$ and  $\iota_V\omega +d H=-\alpha$ for some function $H$ and 1-form $\alpha \in \cI$. Moreover, $\beta \in \cI$ is a 1-form determined by $\alpha$ and $V$ through $\exp(L_V)\omega=\omega+d\beta$.
\end{definition}
Note that $\iota_V\omega +d H=-\alpha$ implies $L_V\omega=d\alpha$ and hence there exists $\beta \in \cI$ such that $\exp(L_V)\omega=\omega+d\beta$. Here we made use of $L_V\cI\subset \cI$ for $V$ vertical. %\jordi{I got a bit lost here.}

Note also that if $\deg_\psi(V)\geq k, k>0$, one finds that  $\cL$ and $\cL^\prime$ coincide up to order $k+2$. At order $k+2$ one gets
\begin{equation}
\cL_k-\cL^\prime_k=Q_0 H_k\,,
\end{equation}
where we made use of $\beta_k=\alpha_k$.
\begin{proposition}
	\label{prop_finite_triv_transf}
	A trivial formal deformation $(E,Q^\prime,\omega,\cL^\prime)$ of a presymplectic gPDE $(E,Q,\omega,\h)$ is again a presymplectic gPDE, in the sense that
	\begin{equation}
		\iota_{Q^\prime}\omega + d\cL^\prime \simeq 0, \qquad Q^\prime \cL^\prime + \frac{1}{2}\iota_{Q^\prime}\iota_{Q^\prime}\omega = 0. \label{eq_prop_finite_triv_transf}
	\end{equation}
    \end{proposition}
	\begin{proof}

Consider the formal transformation $Q \to Q^\prime$, $\cL\to \cL^{\prime\prime}$ with
    \begin{gather}
    Q^\prime=\exp(V) Q \exp(-V)\,, \qquad
    \cL^{\prime\prime}=\exp(V)\cL\,.
    \end{gather}
    Applying $\exp{V}$ to both sides of the defining equations \eqref{iqomega} one gets
    \begin{equation}
    \label{1st-deformed}
       \iota_{Q^\prime}(\omega+d\beta) + d\cL^{\prime\prime} \simeq 0, \qquad Q^\prime \cL^{\prime\prime}+ \frac{1}{2}\iota_{Q^\prime}\iota_{Q^\prime}(\omega+d\beta) = 0\,,
    \end{equation}
where we used  $\iota_{Q^\prime}=\exp({L_V})\iota_Q \exp(-L_V)$ and that $\exp({L_V})\omega=\omega+d\beta$.

Next, by applying $\iota_{Q^\prime}$ to both sides of $\iota_{Q^\prime}d\beta-d \iota_{Q^\prime}\beta=L_{Q^\prime}\beta$
and using that $(Q^\prime)^2$ is vertical one finds $\iota_{Q^\prime}\iota_{Q^\prime}d\beta=2{Q^\prime}(\iota_{Q^\prime}\beta)$. Using this in \eqref{1st-deformed} and introducing $\cL^\prime=\cL^{\prime\prime}+\iota_{Q^\prime}\beta$ we indeed get~\eqref{eq_prop_finite_triv_transf}.
\end{proof}

Let us now get back to the recursive construction of deformations. At order $N+2$, the deformation is determined by the equation $Q_0 \cL_N +\half B_N=0$ along with $\iota_{Q_N}\omega+d\cL_N\in \cI$. The trivial solutions to the first condition are given by $\cL_N=Q_0 f_N$, where $\iota_{V_N}\omega+df_N\simeq 0$ for some vertical vector field $V_N$. It is clear from the above that such an infinitesimal trivial transformation can be achieved by applying a formal deformation determined by $V_N$ which is in turn determined by $f_N$.
More formally, 
\begin{prop}
Let $(E,Q,\omega,\cL)$ and $(E,Q^\prime,\omega^\prime,\cL^\prime)$ be two consistent deformations of the same linear system such that $Q_i=Q^\prime_i$ modulo $\wker$, $\cL_i=\cL^\prime_i$ for $0\leq i<N$ and $\cL_N-\cL^\prime_N=Q_0f_N$ with $f_N$ satisfying the condition that $\iota_{V_N}\omega+df_N\simeq 0$ for some vertical vector field $V_N$, $\deg_\psi V_N=N$. Then applying the trivial formal deformation generated by $V_N$ to $(E,Q,\omega,\cL)$ one finds that it agrees with $(E,Q^\prime,\omega^\prime,\cL^\prime)$ at order $N$ as well. 
\end{prop}
\begin{proof}
By applying the trivial formal deformation determined by $V_N$ one finds that now $\cL_N=\cL^\prime_N$. The only additional point is to check that after the transformation $Q_N=Q^\prime_N$ modulo $\wker$. Indeed, at order $N$ the first equation of \eqref{iqomega} applied to both systems imply $\iota_{Q_N}\omega-\iota_{Q^\prime_N}\omega\in \cI$ which, in turn, implies that $Q_N=Q^\prime_N$ modulo $\wker$.
\end{proof}
The above Proposition shows that at each order the nontrivial arbitrariness of the deformation is described by the cohomology of $Q_0$ in $\algA$.

\iffalse

Consider an infinitesimal trival transformation $\h'_N = \h_N + Q_0 R$. Since $dR$ is in the image by vertical vector fields of $\omega$ modulo $\mathcal{I}$, exists a vertical vector field $V$ of order $N$ such that
\begin{align}
	\iota_V\omega + dR &= \alpha_R \in \mathcal{I}, \label{eq_1st_triv_1} \\
	\mathcal{L}_V\omega &\simeq 0. \label{eq_1st_triv_2}
\end{align}
The previous identities imply that $\gh(V)=0$ and $\deg_\psi(V)=\deg_\psi(R)-2$. Using \eqref{eq_1st_triv_1} we can show that
\begin{equation}
	\h'_N = \h_N + V\h_0 + \iota_{Q_0}\alpha_R.
\end{equation}

The following proposition shows that an infinitesimal trival deformation can be extended to an finite formal trivial\ deformation:

\begin{proposition}
	\label{prop_triv_def}
	Consider the previous transformation $\h'_N = \h_N + V \h_0 + \iota_{Q_0}\alpha_R$. The finite formal trivial\ deformation $(E,\bar{Q},\omega,\bar{\h})$ with
	\begin{equation}
		\label{eq_prop_triv}
		\bar{\h} = \exp(V)(\h + \iota_{Q_0}\alpha_R), \quad	\bar{Q} = \exp(V) Q \exp(-V)
	\end{equation}
	represents it in the sense that $\h' = \bar{\h}$ and $Q' = \bar{Q}$ can be obtained by iterativelly solving \eqref{method_1} and \eqref{method_2} starting by $\h'_{(N)} = \h_{(N-1)} + \h'_N$ and $Q'_{(N-1)} = Q_{(N-1)}$.
	\begin{proof}
		The proof is in Appendix \ref{app_trivialdef}.
	\end{proof}
\end{proposition}
\fi

%%%%%%%%%%%%%%%%%%%%%%%%%%%%%
%%%%%%%%%%%%%%%%%%%%%%%%%%%
\section{Examples}

\label{sec:applications}
\subsection{Chern Simons theory}
%\subsubsection{Linear system}
The starting point is the slight generalization of the Abelian Chern Simons. Namely, a generic linear theory with a 1-form field with values in a linear space and the usual Abelian gauge symmetry. If we denote by $A^i$ the 1-form field, the gauge invariant Lagrangian we have in mind is
\begin{equation}
    L=\half h_{ij}(A^idA^j+A^i\Gamma^j{}_kA^k)\,,
\end{equation}
where $h_{ij}$ is a constant invertible symmetric matrix and $\Gamma^j{}_k$ are fixed 1-forms satisfying $h_{ij}\Gamma^j{}_k+h_{kj}\Gamma^j{}_i=0$ and the zero-curvature equation $d\Gamma+\Gamma^2=0$. The gauge transformations are $\delta A=d\lambda$, where $\lambda$ is an unconstrained gauge parameter (arbitrary spacetime function). Here we restrict ourselves to local analysis and hence refrain from giving a global geometrical construction. Note that one can set $\Gamma=0$ locally by using a suitable transformation of the form $A^i\to A^j\Lambda^i_j(x)$. However, we keep
$\Gamma$ for generality as e.g. in the context of 3d (conformal) gravity models such a transformation is not legitimate as it sets to zero the frame field.

The minimal presymplectic BV-AKSZ formulation of the above system is a slight modification of the well-known AKSZ formulation of the Abelian Chern-Simons theory. The underlying bundle $E\rightarrow T[1]X$ is finite-dimensional with coordinates $ \{ x^a, \theta^a, C^i \} $, where $ \gh{\theta^a} = \gh{C^i} = 1 $, $ \gh{x^a} = 0 $ and $ \dim X = 3 $. The vector field $Q_0$ acts as follows:
\begin{align}
	Q_0 x^a = \theta^a, \quad Q_0 \theta^a = 0, \quad Q_0 C^i = \theta^a \Gamma^i_{aj} C^j.
\end{align}
The compatible symplectic structure is given by:
\begin{equation}
	\omega = d \chi = \half h_{ij}dC^i dC^j,
\end{equation}
where $h_{ij} = h_{ji}$. Finally, the ``covariant Hamiltonian" $\cL_0$ is given by 
\begin{equation}
\cL_0=\half h_{ij} C^i\Gamma^j{}_k C^k\,.
\end{equation}

%\subsubsection{Consistent deformations}
The first-order deformation is given by $\cL_1$ of order $3$ and satisfying 
\begin{equation}
\label{q0e}
	Q_0 \h_1 = 0\,, \qquad \gh{\h_1}=3\,, \qquad \deg_\psi(\h_1)=3\,.
\end{equation}
It follows 
\begin{equation}
	\h_1 = \frac{1}{6}k_{ijk} C^i C^j C^k
\end{equation}
for some  totally antisymmetric $k_{ijk}(x)$. Equation \eqref{q0e} implies that $k_{ijk}$ is covariantly constant with respect to connection $\Gamma$. It's obvious that $ \h_1 $ is not $Q_0$-exact, which means that the deformation is nontrivial. 

Equation \eqref{eq_method_assumptions_2} reads as $\iota_{Q_1} \omega + d \h_1 \simeq 0$
and gives 
\begin{equation}
	Q_1 = -\frac{1}{2} C^i C^j k_{ij}{}^k \dl{C^k},
\end{equation}
where $k_{ij}{}^k = k_{ijl}h^{lk}$. Note that because $h_{ij}$ is invertible $\omega$ is symplectic (seen as a form on the fiber) and hence \eqref{eq_method_assumptions_2} always has a solution. 

In the next step we should look for $\cL_2$ of order $4$ in fields. However, such $\cL_2$ of ghost degree $3$ does not exist and hence, from \eqref{eq_method_assumptions_1}, we get at order 4 the condition $Q_1\cL_1=0$. Because $\omega$ is symplectic, it is a standard fact that this is equivalent to $(Q_1)^2=0$, which in turn implies that $k_{ij}{}^k$ are structure constants (though in our setup they are allowed to depend on $x$) of a Lie algebra, i.e. the Jacobi Identity
\begin{equation}
\label{Jacobi}
k_{ij}{}^l k_{lk}{}^n 
    +
    k_{jk}{}^lk_{li}{}^n
    +
    k_{ki}{}^lk_{lj}{}^n= 0\,.
\end{equation}
This imposes the additional condition on the first-order deformation. Note that $Q_1$ is the Chevalley-Eilenberg differential \eqref{ce_differential} of the Lie algebra defined by structure constants $k_{ij}{}^l$. Finally, it is easy to see that there are no nontrivial conditions at higher orders in fields so that we have arrived at a consistent deformed system. Explicitly:
\begin{align}
	\h &= \h^{(1)} = \h_0 + \h_1 =  \frac{1}{2} h_{ij}C^i\Gamma^j{}_l C^l+ \frac{1}{6}k_{ijk} C^i C^j C^k, \\
	Q &= Q^{(1)} = Q_0 + Q_1 = \theta^a(\pdv{x^a}+\Gamma^i_{aj}C^j\dl{C^i}) -\frac{1}{2} [C,C]^k \dl{C^k},
\end{align}
where $[C,C]^k \coloneqq C^i C^j k_{ij}{}^k$. If there are no auxiliary restrictions like nondegeneracy of the soldering-form component of $\Gamma$ in th case of gravity, one can locally set $\Gamma=0$ and make $k_{ijk}$ constant by the $x$-dependent linear transformation of $C^i$.

The analysis of this section is a generalization of that from~\cite{Barnich:2009jy,Grigoriev:2020lzu} to the case where the initial system is not exactly of AKSZ form.

\subsection{Yang-Mills theory}

%\subsubsection{Linear system}
We start with several copies of Maxwell theory. For simplicity we restrict to 4 dimensions but the generalisation to higher dimensions is straightforward. The minimal presymplectic gPDE describing the free theory is given by $T[1]X \times F$, where $X$ is the Minkowski space with coordinates $x^a$ and $F$ is a graded manifold with coordinates $C^i,F^i_{ab}=-F^i_{ba}$, $\gh{C^i}=1$, $\gh{F^i_{ab}}=0$. The $Q$ vector field  is given by
\begin{align}	
	\label{mmm}		
			Q_0 x^a = \theta^a\,,  \qquad 
			Q_0 \theta^a = 0 \\
			Q_0 C^i = \frac{1}{2} \theta^a \theta^b F^i_{ab}\,, \qquad
			Q_0 F^i_{ab} = 0\,.
\end{align}

The presymplectic structure is given by~\cite{Alkalaev:2013hta,Dneprov:2022jyn,Grigoriev:2022zlq}:
\begin{equation}
	\omega = \theta^{(2)}_{ab} h_{ij} d F^{iab} d C^j \,, \qquad \theta^{(2)}_{ab}=\frac{1}{2}\epsilon_{abcd}\theta^c\theta^d\,,
\end{equation}
where the matrix $h_{ij}$ is assumed symmetric and invertible. The covariant Hamiltonian defined through $\iota_{Q_{0}}\omega+d\cL_0\simeq 0$ is explicitly given by:
\begin{align}
	\cL_0=\frac{1}{4}\left( \theta^{(4)} h_{ij}F^{iab}F^j_{ab} \right)\,,  \qquad \theta^{(4)}=\frac{1}{4!}\epsilon_{abcd}\theta^a\theta^b\theta^c\theta^d\,.
\end{align}
Finally, it is easy to check that $Q_0\cL_0=\iota_{Q_0}\iota_{Q_0}\omega=0$ so that all axioms are satisfied and we have all the data of the linear presymplectic gPDE. Note that $Q_0$ is nilpotent in this case. Further details on the presymplectic BV-AKSZ formulation of YM can be found in~\cite{Grigoriev:2022zlq}.

First order deformations are determined by cocycles $\cL_1$. Taking into account that in addition to 
$Q_0\cL_1=0$, deformation $\cL_1$ must have ghost degree $4$ and order $3$ we have the following candidates:
\begin{itemize}
	\item $\cL_1=\theta^a\theta^b C^i C^j F^k_{ab} k_{ijk} $.\\
    In fact it is trivial. Indeed, $k_{ijk}$ can be assumed antisymmetric in $ij$. At the same time $Q_0\cL_1=0$ implies   that $k_{ijk}$ is antisymmetric in $jk$ as well. This in turn implies that $\cL_1$ can be represented as 
	\begin{equation}
		\theta^a\theta^b C^i C^j F^k_{ab} k_{ijk} = Q_0\left( \frac{2}{3}C^iC^jC^k k_{ijk} \right).
	\end{equation}

	\item $\cL_1= Y_a(x)\theta^a  C^3 $,\\
    which can not be rendered Poincar\'e invariant for nonvansihing $Y_a$.
    
    \item $\cL_1=k_{ijk}(F^i_{ab} F^j_{cd} F^k_{ef})Y^{abcdef} \theta^{(4)}$, where $Y^{abcdef}$ is Lorentz invariant\\
    It is easy to see that it does not involve ghosts and hence does not deform the gauge symmetries. Such deformations are described  by the invariants of the initial linear theory.

    \item $ \cL_1=\half\theta^{(2)}_{ab} C^i C^j F^{kab} k_{ijk} $.
\end{itemize}

Let us concentrate on the last candidate, which is the only interesting one. Note that one can assume $ k_{ijk} = - k_{jik} $. The cocycle condition reads as 
\begin{align}
	Q_0(\cL_1) \propto \theta^{(4)}F^{iab}C^j F^k_{ab} k_{ijk} = 0,
\end{align}
so that $k_{ijk} = -k_{kji}$ and hence $k_{ijk}$ is totally antisymmetric. It is easy to check that $\cL_1$ is not $Q_0$-exact and hence determines a nontrivial first order deformation. Indeed, by counting homogeneity in $F$ and $C$ one finds that it can only arise as $Q_0(C^3)$ but the explicit check shows that it doesn't happen unless $F$ is (anti)self-dual.

The first order deformation of $Q$ is determined by 
\begin{equation}
	\iota_{Q_1}\omega + d\h_1 \simeq 0\,,
\end{equation}
giving
\begin{equation}
	Q_1 =  F^{i}_{ab}C^j k_{ij}{}^k\pdv{F^k_{ab}} -\hhalf C^iC^j k_{ij}{}^k\pdv{C^i},
\end{equation}
where $k_{ij}{}^k \coloneqq k_{ijl} h^{lk}$.

The next order deformation is determined by  $\h_2$ with $ \deg_\psi (\h_2) = 4 $ and $ \gh{\h_2} = 4 $ such that
\begin{equation}
	Q_0 \h_2 = -\frac{1}{2} Q_1 \h_1  \label{eq_max_i2}\,.
\end{equation}
This equation says that
\begin{equation}
	Q_1 \h_1 = \theta_{ab}^{(2)} C^i C^j C^k F^{lab}k_{ij}{}^m k_{mkl}, \label{eq_ym_Q1L1}
\end{equation}
is trivial in the cohomology of $Q_0$.  The only possible candidates for $\cL_2$ are of the form $ \theta^a Y_a C^3 F $ and $ C^4 $. The first cannot be rendered Poincar\'e invariant. At the same time $Q_0\cL_2$ with 
\begin{equation}
	\cL_2 =C^iC^jC^kC^l q_{ijkl}
\end{equation}
%$q_{ijkl}$ totally 
can not coincide with $Q_1\cL_1$ unless $F$ is (anti)self-dual. This shows that  $Q_1\cL_1$ represents an obstruction.
In fact, $Q_1\cL_1$ vanishes if and only if $k_{ij}{}^k$ satisfies the Jacobi identity \eqref{Jacobi}. In this way we recover the presymplectic gPDE for the standard YM theory~\cite{Grigoriev:2022zlq,Dneprov:2022jyn}:
\begin{align}
	Q &= \theta^a \pdv{x^a} + \frac{1}{2}\theta^a\theta^b F^i_{ab}\pdv{C^i} + [F_{ab},C]^i\pdv{F^i_{ab}} - \frac{1}{2} [C,C]^i\pdv{C^i} \label{result_maxwell_Q},  \\
	\h  &= \frac{1}{4}\theta^{(4)}F^i_{ab} F_i^{ab} + \hhalf  \theta^{(2)}_{ab} F_i^{ab}  [C,C]^i  \label{result_maxwell_L},
\end{align}
where $[A,B]^k \coloneqq A^i B^j k_{ij}{}^{k}$.

\section{Conclusions}

In this work we have developed the deformation theory for local gauge theories within the  presymplectic BV-AKSZ approach. This has allowed us to reformulate the analysis of consistent deformations in terms of finite-dimensional super-geometrical objects -- the so-called presymplectic gauge PDEs. The approach has substantial technical advantages because it avoids working with quotient spaces such as local functions modulo total derivatives and modulo on-shell vanishing terms. One can also draw an analogy with the Hamiltonian formalism which is also ``on-shell''. More precisely, the presymplectic BV-AKSZ approach can be thought of as the BV-BRST extension of the covariant phase space formalism~\cite{Lee:1990nz,Wald:1999wa} phrased in terms of the geometrical theory of PDEs~\cite{Vinogradov1981,Krasil?shchik-Lychagin-Vinogradov}.

In order to simplify the discussion, we restricted ourselves to the expansion in homogeneity in fields, as this is the standard setup for the analysis of consistent interactions. Of course, the generalization to the expansion in general deformation parameter (e.g. coupling constant) is straightforward. 

It is well known that the deformation theory of gauge systems can be employed in the study of perturbative renormalization, see, e.g.~\cite{Piguet:1995er,Barnich:2000me}. The same should apply to the framework proposed in the present work. However, the only subtlety is the treatment of higher-derivative counterterms that would require certain modifications of the procedure.

Possible applications of the approach involve the analysis of the interactions of higher spin fields and massive fields. Another promising direction is the deformations of gravity-like theories for which the presymplectic BV-AKSZ formulation is especially concise, see e.g.~\cite{Grigoriev:2020xec,Dneprov:2022jyn,Dneprov:2024cvt,Grigoriev:2024ecv}. However, that would require taking into account deformations of the presymplectic form as this structure is usually field-dependent in the case of gravity-like models.

%%%%%%%%%%%%%%%%%%%%%%%%%%%%%%%%%%%%%%%%%%%%%%%%%%%
%%%%%%%%%%%%%%%%%%%%%%%%%%%%%%%%%%%%%%%%%%%%%%%%%%%

\section*{Acknowledgments}
\label{sec:Aknowledgements}
%%%%%%%%%%%%%%%%%%%%%%%%%%%%%%%%%%%%%%%%%%%%%%%%%%%%%%%%%%%%%
We are grateful to I.~Dneprov and A.~Kotov for useful discussions. 
M.G. wishes to thank  M.~Henneaux, E.~Skvortsov, and especially G.~Barnich for fruitful exchanges.

%%%%%%%%%%%%%%%%%%%%%%%%%%
%%%%%%%%%%%%%%%%%%%%%%%%%%
\setlength{\itemsep}{1pt}
\small

%\bibliography{HSmaster}
\providecommand{\href}[2]{#2}\begingroup\raggedright\endgroup

\end{document}